\documentclass[12pt]{amsart}

\usepackage{graphicx}
\usepackage{pstricks}

\newtheorem{thm}{Theorem}[section]

\newtheorem{lemma}[thm]{Lemma}
\newtheorem{conj}[thm]{Conjecture}
\newtheorem{cor}[thm]{Corollary}

\newtheorem{example}[thm]{Example}

\usepackage{amsmath}
\usepackage{amsxtra}
\usepackage{amscd}
\usepackage{amsthm}
\usepackage{amsfonts}
\usepackage{amssymb}
\usepackage{eucal}

\newcommand{\ty}{{\tilde y}}

\newcommand{\C}{{\mathbb C}}
\newcommand{\Z}{{\mathbb Z}}

\newcommand{\al}{{\alpha}}

\numberwithin{equation}{section}

\begin{document}
\title[Discrete non-commutative integrability]{Discrete non-commutative integrability: the proof of a conjecture by M. Kontsevich}

\author{Philippe Di Francesco} 
\address{PDF: Institut de Physique Th\'eorique du Commissariat \`a l'Energie Atomique, 
Unit\'e de Recherche associ\'ee du CNRS,
CEA Saclay/IPhT/Bat 774, F-91191 Gif sur Yvette Cedex, 
FRANCE. e-mail: philippe.di-francesco@cea.fr}

\author{Rinat Kedem}
\address{RK: Department of Mathematics, University of Illinois Urbana, IL 61801, 
U.S.A. e-mail: rinat@illinois.edu} 
\date{\today}
\begin{abstract}
We prove a conjecture of Kontsevich regarding the solutions of rank two recursion 
relations for non-commutative variables which, in the commutative case, reduce to 
rank two cluster algebras of affine type. 
The conjecture states that solutions are positive Laurent polynomials in the initial 
cluster variables. We prove this by use of a non-commutative version of
the path models which we used for the commutative case. 
\end{abstract}

\maketitle
%\tableofcontents

\section{Introduction}

Let $\mathbb F=\C(x,y)$ denote the skew field of rational functions in the 
non-commutative variables $x$ and $y$. 
Given any $a\in\Z$, Kontsevich introduced the following transformation on $\mathbb F^2$:
\begin{equation}\label{taop}
T_a\ : \ \begin{pmatrix} x\\ y \end{pmatrix} \ \mapsto \ 
\begin{pmatrix} x y x^{-1}\\ (1+y^a)x^{-1}. \end{pmatrix}
\end{equation}
which preserves the commutator $C=xyx^{-1}y^{-1}$.

Let ${\mathcal A}\subset \mathbb F$ be the algebra generated by the entries of all 
vectors obtained from iterations of the map $T_c T_b$ (with $c,b\in\Z$), 
acting on the vector 
$\begin{pmatrix} x\\y \end{pmatrix}$. Kontsevich made the following conjecture:
\begin{conj}\cite{Kont}\label{main}
For any $b,c\in \Z_{>0}$, the entries of the vector 
$(T_cT_b)^m\begin{pmatrix} x\\ y \end{pmatrix}$, for all
$m\geq 0$, are non-commutative Laurent polynomials in $x$ and $y$ 
with non-negative integer coefficients. 
\end{conj}
That is, the generators of $\mathcal A$ are positive Laurent polynomials. 

This conjecture is analogous to the similar conjecture  \cite{FZ} for rank 2 cluster algebras, 
in the commutative limit, $C=1$.
Although it is not quite clear what a ``good'' definition of a noncommutative cluster 
algebra should be in general, the equation \eqref{taop} can be thought of an example 
of the mutations of the cluster variables in a rank 2 non-commutative cluster algebra.

In the commutative case \cite{SZ}, we introduced a method \cite{cluster4} which guarantees 
Laurent positivity, valid for the integrable cases of rank 2 cluster algebras, 
corresponding to a rank 2 affine Cartan matrix (that is, $b c = 4$).
This was done by writing explicit expressions for the cluster variables in terms 
of path models on weighted graphs. Equivalently, the generating function for 
cluster variables is a finite continued fraction with a manifestly positive expansion.  

The case of affine $A_1$ is also the simplest example of a $Q$-system cluster algebra 
\cite{Kedem}. The path formulation can be generalized to the higher rank $Q$-systems 
\cite{cluster3,cluster4}. 
A non-commutative version of the $Q$-system cluster algebra is given by 
the so-called $T$-system. In \cite{cluster5}, we found the solutions of the $A_r$ $T$-system 
using the same path models, but with non-commutative weights. In this case, the 
non-commutative mutation relations are such that their matrix elements are the 
$T$-system equations, and the matrix elements of the non-commutative cluster variables 
are  the $T$-system solutions. This is therefore another candidate for a non-commutative 
cluster algebra, of higher rank.

At rank 2, the Kontsevich evolution and the non-commutative $Q$-system relations are 
candidates for non-commutative cluster algebra mutations.
Both can actually be obtained as different specializations 
of a more general evolution equation. 
We remark that the non-commutative $Q$-system 
is distinct from the quantum cluster algebras defined by \cite{BZ}, which is obtained as a 
specialization of the Kontsevich evolution, by setting $C=q$ to be a central element. 

In this paper, we prove the conjecture of Kontsevich in the cases where it generalizes 
the integrable rank 2 cluster algebras, that is, the values of $b,c$ are obtained from 
an affine Cartan matrix. We again use the path models introduced in \cite{cluster4} 
with non-commutative weights.

\noindent{\bf Acknowledgements:} We thank M. Kontsevich for explaining his 
conjectures to us. The research of P.D.F. is supported in part by the 
ANR Grant GranMa, the
ENIGMA research training network MRTN-CT-2004-5652,
and the ESF program MISGAM. 
R.K. is supported by NSF grant DMS-0802511. 
This research was hosted by the Mathemathisches 
Forschungsingstitut Oberwolfach and by the IPhT at CEA/Saclay. 
We thank these institutes for their support.

\section{Preliminaries}
\subsection{Path models for the affine rank 2 cluster algebras}
Let us review briefly our solution for the rank 2 affine cluster algebras in the 
commutative case \cite{FZ}. For the full details we refer to \cite{cluster4}.

Let $\mathcal F$ be the field of rational functions in the two commuting formal 
variables $x,y$ with rational coefficients. We consider the subring of $\mathcal F$ 
generated by the variables $R_n$, where $R_n$ satisfy the recursion relations 
\begin{equation}\label{commutativeexchange}
R_{n+1}R_{n-1} = \left\{ \begin{array}{ll} 1 + R_n^b, & n \hbox{ odd};\\
1 + R_n^c, & n \hbox{ even},\end{array}\right.
\end{equation}
with initial conditions $(R_0,R_1)=(x,y)$. This is the cluster algebra of rank 2 
corresponding to the exchange matrix 
\begin{equation}\label{Bmat}
B=\begin{pmatrix} 0 & b \\ -c & 0\end{pmatrix}.\end{equation}
with $b,c>0$. 

Some of these cluster algebras were studied in \cite{SZ,CZ}, where, in particular,  
Laurent positivity was proven for the cases $bc\leq 4$ and $b=c$ in general. 
We are interested here in the {\em integrable} cases of the discrete evolution 
equations \eqref{commutativeexchange}, that is, 
 the following cases: (i) $(b,c)=(2,2)$, (ii) $(b,c)=(1,4)$ and (iii) $(b,c) = (4,1)$. 
Each of these cases corresponds to an affine Dynkin diagram of rank 2, and 
each is integrable. These cases were studied in \cite{SZ}.

We may write a map $T_a:{\mathcal F}^2\to {\mathcal F}^2$ as
$$
T_a \begin{pmatrix} R_n \\ R_{n+1} \end{pmatrix} = 
\begin{pmatrix} R_{n+1} \\ R_{n+2}\end{pmatrix}.
$$
Then for th compound transformation $\mu=T_c T_b$ we have
$\mu^m (R_0,R_1)=(R_{2m},R_{2m+1})$, so there is a translational property, 
$\mu^m(R_{2n},R_{2n+1})=(R_{2(m+n)},R_{2(m+n)+1})$. Therefore to get an 
expression for  $R_n$ in terms of any initial data $(R_m,R_{m+1})$, we need 
only find it in terms of two sets of initial data, say, $(R_0,R_1)$ and $(R_1,R_2)$. 
If $b=c$ we need only consider the first set of data because of the additional symmetry.

The path model solution of the system goes as follows. 
Define the generating function for the cluster variables as follows:
\begin{equation}
F(t) = \left\{ \begin{array}{ll} \sum_{n\geq 0} t^n R_n & \hbox{case (i)};\\
\sum_{n\geq 0} t^n R_{2n} & \hbox{case (ii)};\\
\sum_{n\geq 0} t^n R_{2n+1} & \hbox{case (iii)}.\end{array}\right.
\end{equation}
We write explicit expressions for $F(t)$ in each case, which make manifest the 
property that the coefficients of $t^n$ are positive Laurent polynomials in any 
seed cluster data. One can then use the symmetries of these systems to express 
all other cluster variables (for negative values of $n$ and for odd/even values of 
$n$ in the cases (ii) and (iii)) in terms of the coefficients of $F(t)$. The expressions 
are always such that if the coefficients of $F(t)$ are positive Laurent polynomials, 
so are the remaining cluster variables.

One way of computing the generating function $F(t)$ is by expressing it as the 
partition function of weighted paths on a graph. If the weights are positive 
Laurent polynomials in some initial seed data, then so are the coefficients 
of $t^n$ in $F(t)$, and hence the cluster variables. 

For a given graph with vertices connected by edges, we assign a weight 
$w_{i,j}$ ($i,j$ two vertices) associated to the edge connecting vertex 
$i$ to vertex $j$. In general, we do not require that $w_{ij}$ is equal to $w_{ji}$. 
The weight of a path from vertex $i$ to vertex $j$ is the product of all the 
weights along the edges traversed by the path, and its length is the number 
of steps traversed. The partition function of paths from vertex $a$ to vertex 
$b$ is the sum over all paths from vertex $a$ to vertex $b$ of the weights of the paths.

One can use the following presentations for the generating functions $F(t)$ in 
terms of path partition functions.
\begin{itemize}
\item $(b,c)=(2,2)$: Consider the graph composed of 4 vertices, labeled $0,1,2,3$, 
with edges connecting vertices $i$  and $i+1$. We assign a weight $w_{ij}$ to a 
step from $i$ to $j$ along each edge, when such an edge exists. 
The weights are $w_{i,i+1}=1,\ w_{i,i-1}=t y_i$, 
where 
$$ y_1=R_1 R_0^{-1},\ y_2 = R_1^{-1}R_0^{-1},\ y_3=R_1^{-1}R_0. $$ 
Note that these weights are positive Laurent monomials in the initial data $R_0,R_1$. 

Then one can show that $F(t)R_{0}^{-1}$ is equal to the partition function of paths 
from vertex 0 to itself on this weighted graph. Since each weight is a positive 
Laurent monomial in the initial data, so is each coefficient of $t^n$ in the partition 
function, which represents paths of length $2n$. 

An equivalent statement is that the generating function $F(t)$ is the following finite 
continued fraction:
$$F(t) R_0^{-1} = \frac{1}{1-t \frac{y_1}{1-t\frac{y_2}{1-ty_3}}}.
$$

We remark that a key fact which enables us to prove these formulas for the solution 
is the existence of  an integral of motion, 
$K=y_1 + y_2 + y_3$, and also that $y_3 y_1=1$. $K$ is invariant 
under $T_2$, namely under $R_n\mapsto R_{n+1}$ for all $n$.  
\begin{figure}
\centering
\includegraphics[width=12.cm]{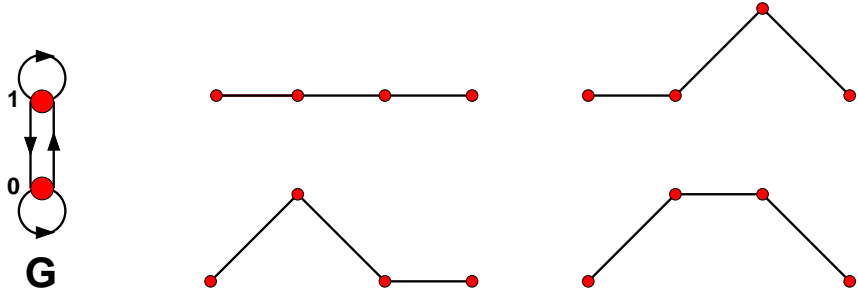}
\caption{\small Weighted paths on the barbell graph $G$ generate $R_n$ 
for the cluster algebras $(1,4)$ and $(4,1)$. 
We have represented the $4$ paths of length $3$ on $G$, from the vertex $0$ to itself,
with respective weights $y_1^3$, $y_1y_2$, $y_2y_1$ and $y_3y_2$.}\label{fig:barbell}
\end{figure}
\item $(b,c)=(1,4)$: Let $G$ be the barbell graph, with vertices $0$ and $1$ 
connected by an edge, and each vertex connected to itself by a loop. 
See Figure \ref{fig:barbell}. We consider paths on this graph from node 0 
to itself of length $n$. We place a weight on each oriented edge. A path 
step from $i$ to $j$ contributes a weight $w_{i,j}$, so that a step from 1 to 
1 around the loop connected to it contributes weight $w_{1,1}$. We choose
$
w_{0,0}=t y_1,\
w_{0,1}=t,\
w_{1,1}=t y_3,\
w_{1,0}=t y_2,
$
where 
$$
y_1= \frac{1+R_1}{ R_0},\ y_2 = \frac{R_0^4 +(1+ R_1)^2}{R_0^4 R_1},\ y_3 = 
\frac{R_0^4 + 1 + R_1}{R_0^2 R_1}.
$$
This time, the weights are not monomials, but they are positive Laurent polynomials. 
Paths on this graph therefore all have weights which are positive Laurent polynomials. 
Explicitly, the generating function for $R_{2n}$ in terms of the cluster seed $(R_0,R_1)$ 
is obtained from the expansion in $t$ of the finite continued fraction
$$
F(t) = \frac{R_0} {1- ty_1 - t^2 y_2\frac{1}{1-t y_3}}.
$$
The odd cluster variables are obtained from the relation $R_{2n}R_{2n-2}-1=R_{2n-1}$. 
One can show that $1$ is a term in the Laurent polynomial $R_{2n} R_{2n-2}$ so that 
$R_{2n-1}$ is positive.
We note here that we have a conserved quantity, $K=y_1+y_3$, and that 
$y_1 y_3-y_2=1$. $K$ is left invariant under the compound 
mutation $\mu$, i.e. under $R_n\mapsto R_{n+2}$ for all $n$.

Positivity with respect to the initial data $(R_1,R_2)$ follows from the solution 
to the problem with $(b,c)=(4,1)$, due to the symmetry between the two systems.

\item $(b,c)=(4,1)$: Here, we use the same graph, but replace the weights $y_i$ 
with $y'_{\al}(R_0,R_1)= y_{4-\al}(R_1,R_0)$. The generating function for odd 
cluster variables is the partition function on this graph of paths from node 1 to itself. 
Even cluster variables are obtained from the equation $R_{2n+1} R_{2n-1}-1=R_{2n}$, 
and can again be shown to be positive. The conserved quantity is $K'=y_1'+y_3'$, 
and we also have $y_3'y_1'-y_2'=1$. 
\end{itemize}

We will show below that each of these results generalizes to the non-commutative 
case. We will have more conserved quantities in the non-commutative case, due 
to the fact that the commutator $C\neq 1$, and is not, in fact, central. Thus, $C$ 
itself will be a conserved quantity.

\subsection{A rank 2 non commutative cluster algebra}\label{ranktwo}
We consider now the evolution \eqref{taop}.
Define $C=x y x^{-1} y^{-1}$ and let $R_0 = C x$ and $R_1=y$. Given a pair 
$b,c\geq 0$, define $\{R_n\}_{n\in \Z\geq 0}$ by
$$
T_a \begin{pmatrix} C R_n \\ R_{n+1} \end{pmatrix} = \begin{pmatrix} C 
R_{n+1} \\ R_{n+2} \end{pmatrix}
$$
where $a=b$ if $n$ is even and $a=c$ if $n$ is odd. That is, 
\begin{equation}\label{mutation}
R_{n+1} C R_{n-1} = \left\{\begin{array}{ll} 1 + R_{n}^b, & n \hbox{odd};\\
1+R_n^c, & n\hbox{ even}.\end{array}\right.
\end{equation}
Clearly, this defines $R_n$ for negative values of $n$ as well.

The expression for the commutator, 
\begin{equation}\label{commutator}
C=R_{n+1}^{-1} R_n R_{n+1} R_n^{-1},
\end{equation}
allows us to interpret it as a  conserved quantity of the discrete evolution 
$T_a$, because its value is 
independent of $n$. 
In fact, one can check that, generally, any recursion 
relation of the form
\begin{equation}\label{comr}
R_{n+1} R_n^{-1} R_{n-1} = f_n(R_n)
\end{equation}
has $C$ as a conserved quantity. Equation \eqref{mutation} is a special 
case of this. 

Note also that Eq. \eqref{commutator} implies a quasi-commutation relation 
\begin{equation}\label{qcomm}
R_{n+1} C R_n = R_n R_{n+1}. 
\end{equation}

If $C=1$, that is, if $x$ and $y$ commute, we recover the rank $2$ cluster 
algebra of type $(b,c)$.

Similarly, if we write $C=q$, a central element, then Eq. \eqref{qcomm} turns into the
quantum commutation relation $R_n R_{n+1} = q R_{n+1} R_n$, 
and we recover the rank 2 quantum cluster algebra of \cite{BZ}.

Therefore we call the transformation \eqref{taop} a mutation and the ring 
$\mathcal A$ a non-commutative cluster algebra of rank 2. In general, the 
Laurent property is not proven for this algebra except in special cases. 
For example, in the case where $B$ is obtained from the Dynkin diagram 
of finite type $A_2, B_2$ or $G_2$, one can check that the cluster algebra 
is finite, up to conjugation by $C$, with the same period as in the commutative 
case. In those cases, the cluster variables are positive Laurent polynomials 
with coefficients which are either $0$ or $1$. In the case where $B$ is 
obtained from the Dynkin diagram of affine $A_1$, the Laurent property 
has been proved by Usnich, but not the positivity \cite{Kont}.

In this paper we do not attempt a general proof of Conjecture \ref{main}, but 
we generalize our proof for the commutative integrable (affine Dynkin diagram) 
cluster algebras of rank 2. We show that the path models  of 
Ref.\,\cite{cluster4} have a simple non-commutative analogue. We therefore get 
an explicit expression for all cluster variables $R_n$. This allows us to prove 
the Laurent property and positivity for those cases.

What distinguishes these cases is that the transformation \eqref{mutation} is 
{\em integrable}.  That is, in each case, there exist two Laurent polynomials 
(one of them being the commutator $C$) in the variables  $R_n, R_{n+1}$, 
which are invariant under \eqref{mutation}. The generating function for the cluster 
variables can be expressed in all cases as a finite continued fraction in any 
cluster seed variables. This can be interpreted as a generating function for 
paths with non-commutative weights.

\subsection{Symmetries}\label{symms}

The evolution equations \eqref{mutation}
determine all $R_n$ for $n\in \Z$ uniquely in terms of the initial data
$CR_0=x$ and $R_1=y$.

One may relate the solutions of the $(c,b)$ system to those of the 
$(b,c)$ system by use of the translational symmetry.
Let us  denote by $f^{(b,c)}_n(x,y)$
the solution $R_n$ of \eqref{mutation} expressed in terms of its initial data
$CR_0=x$ and $R_1=y$. Let us also denote by $g^{(b,c)}_n(X,Y)$ the
solution $R_n$ of \eqref{mutation} expressed in terms the data $CR_1=X$, $R_2=Y$.

\begin{lemma}\label{bctrans} For all $n\in \Z$, we have:
$$f^{(c,b)}_n(x,y)= g^{(b,c)}_{n+1}(x,y), \ n\in\Z.$$
\end{lemma}
Thus, solutions of the $(c,b)$ system are given by those of the $(b,c)$ system. 

Moreover, one can 
relate the solution for $n<0$ to that for $n\geq 0$. Define an anti-automorphism $*$ on $\mathbb F$ by 
\begin{equation}\label{defstar}
x\mapsto x^*=y C=yxyx^{-1}y^{-1},\quad y\mapsto y^*=C^{-1} x=y x y^{-1}.
\end{equation}
This is clearly an involution. 
In particular, we have $C^*=C$, $R_1^*=R_0$ and $R_0^*=R_1$. 

Let $R_n=f^{(b,c)}_n(x,y)$ be a solution of the $(b,c)$-system, and
$S_n=R_{1-n}^*$. Changing $n\to 1-n$ in \eqref{mutation} and applying $*$, 
we see that $S_n$ satisfies the $(c,b)$-system. Since
$S_0=R_1^*=R_0$, so that $CS_0=x$, and $S_1=R_0^*=R_1=y$, we have 
that $S_n=f_n^{(c,b)}(x,y)$. Therefore,
\begin{equation}\label{symthm}
f_{-n}^{(c,b)}(x,y)=\left(f_{n+1}^{(b,c)}(x,y)\right)^*.
\end{equation}

Note that the anti-automorphism $*$ \eqref{defstar} sends positive 
Laurent monomials of $x,y$ 
to positive Laurent monomials of $x,y$. 

To summarize, the symmetries assure us that, in the case $b=c=2$, it is
sufficient to prove that $R_n$ is a positive Laurent polynomial of $x,y$ for $n\geq 0.$
For $(b,c)=(1,4)$, we may restrict our attention to $n\geq 0$
but we must find  $R_n$ as a function of both
$(x,y)=(CR_0,R_1)$ and $(X,Y)=(CR_1,R_2)$. For $(b,c)=(4,1)$, 
the solutions will be expressed in terms of those of $(1,4)$. 
In all cases, the expressions for $n<0$ follow from equation \eqref{symthm}.

\section{The non-commutative cluster algebra in the case $b=c=2$}\label{twotwo}

\subsection{Conserved quantities and linear recursions}

The non-commutative $(2,2)$-system
\begin{equation}\label{ncq}
R_{n+1}CR_{n-1}=(R_n)^2+1
\end{equation}
is a discrete integrable equation in the sense that it has a 
conserved quantity in addition to the commutator $C$ \eqref{commutator}. 
\begin{lemma}\label{conster}
The polynomial in the solutions $R_n$ of the $(2,2)$ system \eqref{mutation}
 $K=R_{n+1}^{-1} R_n +R_{n+1}^{-1} R_n^{-1}+R_{n+1}R_n^{-1}$
 is independent of $n$.
\end{lemma}
\begin{proof}
Define
\begin{equation} \label{linKL}
K_n=R_n^{-1}(R_{n+1}C+R_{n-1}), \qquad L_n=(R_{n+1}+CR_{n-1})R_n^{-1}\ .
\end{equation}
Then $K_n=L_n$ as a consequence of the first conservation law, upon substituting
$C=R_{n+1}^{-1}R_nR_{n+1}R_n^{-1}$ into the expression for $K_n$ and 
$C=R_n^{-1}R_{n-1}R_nR_{n-1}^{-1}$ into that for $L_n$. 
Subtracting Equation \eqref{ncq} for $n$ from that for $ n+1$,
$$0=(R_{n+2}CR_n-R_{n+1}^2)-(R_{n+1}CR_{n-1}-R_n^2)=R_{n+1}(K_{n+1}-L_n)R_n$$
So we deduce that $K_n=K$ is independent of $n$.
\end{proof}

Substituting  $L_n=K_n=K$ into \eqref{linKL}, we have

\begin{lemma}
There exist two linear recursion relations with constant coefficients satisfied by 
the solutions of \eqref{mutation}:
\begin{eqnarray}
R_{n+1}C+R_{n-1}&=&R_n K \label{recur}\\
R_{n+1}+CR_{n-1}&=&KR_n  \label{recul}
\end{eqnarray}
\end{lemma}
These two recursion relations are equivalent modulo the first conserved quantity.

\subsection{Paths with noncommutative weights and positivity}\label{twotwopath}

Define a generating function for the variables $R_n$ with $n\geq 0$,
$$F(t)=\sum_{n\geq 0} t^n R_n.$$
\begin{thm}\label{continued}
\begin{equation}\label{contfrac}
F(t)=\left(1-t \left(1-t(1-t y_3)^{-1} y_2\right)^{-1} y_1\right)^{-1} R_0,
\end{equation}
where the "weights" $y_i$ are defined as
\begin{equation}
 y_1=R_1R_0^{-1},
\ y_2=R_1^{-1}R_0^{-1}, \ y_3=R_1^{-1}R_0.
\end{equation}
\end{thm}
\begin{proof}
Using Equation \eqref{recul},
$$
F(t)=(1-t K+t^2 C)^{-1} (R_0-t(KR_0-R_1))
$$
Noting that $K=R_1R_0^{-1}+R_1^{-1}R_0^{-1}+R_1^{-1}R_0=y_1+y_2+y_3$, 
$K-R_1R_0^{-1}=y_2+y_3$, and $C=y_3y_1$, we have 
\begin{eqnarray*}
F(t)&=&(1-t(y_1+y_2+y_3)+t^2 y_3y_1)^{-1}(1-t (y_2+y_3))R_0\\
&=& \left(1-t (1-t(y_2+y_3))^{-1}(1-t y_3)y_1\right)^{-1}R_0\\
&=& \left(1-t \left(1-t(1-ty_3)^{-1}y_2\right)^{-1}y_1\right)^{-1}R_0
\end{eqnarray*}
and the Theorem follows.
\end{proof}

This expression for $F(t)$ is to be considered as a power series in $t$ with 
coefficients which are words in the non-commutative variables $y_1,y_2, y_3$. 
Substituting
\begin{equation}\label{positvals}
y_1=y^2x^{-1}y^{-1}, \ y_2=x^{-1}y^{-1}, \ y_3=x y^{-1}, \ R_0=yxy^{-1}
\end{equation}
into \eqref{contfrac}, we deduce
\begin{cor}\label{positcor}
For all $n\geq 0$, the solution $R_n$ of Equation \eqref{ncq} is a Laurent 
polynomial of $x,y$ (with $x=CR_0$ and $y=R_1$)
with only non-negative integer coefficients.
\end{cor}
From the discussion of the previous section, the same is true for $R_n$ with $n<0$.

We can use Equation \eqref{contfrac} to interpret $R_n$ as a partition
function for paths with non-commuting weights. 
\begin{thm}\label{paths}
For all $n\geq 0$, the quantity $R_nR_0^{-1}$,
where $R_n$ is the solution 
\ref{ncq}, is the partition function for paths along the segment $[0,3]$
starting and ending at $0$ with $2n$ steps, with a weight $1$ per step $i\to i+1$
and $y_i$ per step $i\to i-1$ given by \eqref{positvals}, the total (non-commutative)
weight of each path being the product from left to right of
the step weights in the order in which they are visited. 
\end{thm}
\begin{proof}
The continued fraction $F(t)R_0^{-1}$
of Theorem \ref{continued} may be computed 
by the following recursion:
\begin{eqnarray*}
F_k&=&(1-t F_{k+1}y_k)^{-1}  \ \ (k=1,2,3)\\
F_4&=&1\\
F(t)R_0^{-1}&=&F_1
\end{eqnarray*}
To get the series in $t$, we have to expand each intermediate step as:
$$F_k=\sum_{n\geq 0} t^n (F_{k+1} y_k)^n$$
Using this as an induction step, it allows to interpret $F_k$
as the partition for paths on $[k-1,3]$, from and to $(k-1)$, 
with weight $1$ per step $i\to i+1$ and $ty_i$ per step $i\to i-1$. 
This is clearly true for $F_4=1$, the partition function for the trivial path
from $3\to 3$ on the set $\{3\}$, with zero step.
For intermediate $k$'s, we simply decompose paths on $[k-1,3]$ from $k-1\to k-1$
into segments delimited by the ascending steps $k-1\to k$
and the next descending step $k\to k-1$,
the former receiving the weight $1$ the latter the weight $ty_k$. 
In-between any two such steps, the path only explores the segment $[k,3]$, 
with the partition function $F_{k+1}$. Finally,
the weights are multiplied in the same order in which the steps are taken, 
and the Theorem follows.
\end{proof}

\begin{figure}
\centering
\includegraphics[width=15.cm]{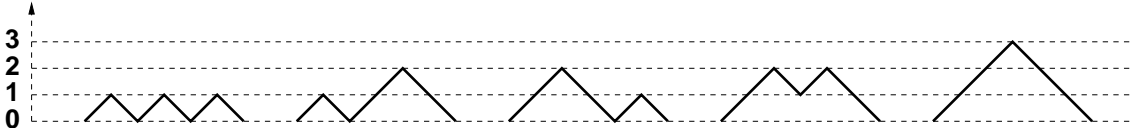}
\caption{\small The five paths on $[0,3]$ of length $6$, from $0\to 0$.}\label{fig:fivepath}
\end{figure}

\begin{example}
For $n=3$, the five paths on $[0,3]$ with $6$ steps from $0\to 0$ are depicted in Figure
\ref{fig:fivepath}, and contribute respectively to $R_3$ (with weight $y_i$ per
descending step $i\to i-1$):
\begin{eqnarray*}
R_3&=&(y_1^3+y_1y_2y_1+y_2y_1^2+y_2^2y_1+y_3y_2y_1)R_0\\
&=&y^2x^{-1}yx^{-1}+y^2x^{-1}y^{-1}x^{-1}+x^{-1}yx^{-1}+x^{-1}y^{-1}x^{-1}+x y^{-1}x^{-1}\\
&=&\Big(\big((1+y^2)x^{-1}\big)^2+1\Big)xy^{-1}x^{-1}=(R_2^2+1)R_1^{-1}C^{-1}
\end{eqnarray*}
\end{example}

An alternative formulation uses
the following transfer matrix, with non-commutative entries indexed by $0,1,2,3$:
\begin{equation}
T=\begin{pmatrix} 0 & 1 & 0 & 0\\
t y_1 & 0 & 1 & 0\\
0 & t y_2 & 0 & 1\\
0 & 0 & t y_3 & 0 \end{pmatrix}
\end{equation}
The matrix element $T_{i,j}$ is nothing but the non-commutative weight of the step $i\to j$
of the above paths. Consequently, $(T^n)_{i,j}$ is the partition function for paths of 
$n$ steps, from $i$ to $j$. The partition function for paths from $0\to 0$ is therefore
$$
F(t)R_0^{-1}=\Big(\sum_{n=0}^\infty T^n \Big)_{0,0} =\Big(( I - T)^{-1}\Big)_{0,0}.
$$
There is a direct link between this formulation and 
the continued fraction expression \eqref{contfrac}. The latter is obtained by Gaussian elimination of the matrix $I-T$, by allowing only
multiplication from the left and addition of rows (left multiplication by an upper triangular matrix), resulting in a
lower triangular matrix. $F(t)R_0^{-1}$ is then computed as the inverse of the first diagonal element in the resulting matrix,
leading to equation \eqref{contfrac}. 

Both formulations display explicitly the
positivity of $R_n$ as a Laurent polynomial of $x,y$. We have

\begin{thm}\label{antiauto}
For any $n\in \Z$, the solution $R_n$ of  \eqref{ncq}
is a Laurent polynomial of $x,y$ with non-negative integer coefficients $\in \{0,1\}$. Each Laurent monomial in the expression for $R_n$ corresponds to the weight of a single path on the segment $[0,3]$ with weights as in Theorem \ref{paths}.
\end{thm}
\begin{proof}
The expression for $R_n$, $n\geq 0$, as the path
partition function of Theorem \ref{paths} times $R_0$ is a manifestly positive
Laurent polynomial of $x,y$. From \eqref{symthm}, we deduce the positive
Laurent polynomiality for all $n\in\Z$.

Moreover, 
as the weights $y_i$ all have the form 
$R_1^{\pm1} R_0^{\pm1}$, they do not commute with each-other. 
Each path contributes an ordered 
product of such weights from left to right, 
which is encoded in the succession of down steps (taken towards the
origin) along the path. This gives a bijection between the paths contributing to $R_n$
and their total weights, which therefore occur exactly once in the expression
of $R_nR_0^{-1}$. The last part of the Theorem follows.
\end{proof}

\section{The cases $(b,c)=(1,4)$ and $(4,1)$.}\label{onefour}

The non-commutative $(1,4)$ recursion relations \eqref{mutation} can be written as:
\begin{eqnarray}\label{qsysonefour}
R_{2n}C R_{2n-2}&=&1+R_{2n-1}\nonumber \\
R_{2n+1}CR_{2n-1}&=&1+(R_{2n})^4 
\end{eqnarray}
As explained in Section \ref{symms}, we consider two different sets of
initial conditions: $(x,y)=(CR_0,R_1)$, where
\begin{equation}\label{initio}
R_0=yxy^{-1}, \qquad R_1=y
\end{equation}
and $(X,Y)=(CR_1,R_2)$, where
\begin{equation}\label{initb}
R_1=YXY^{-1}, \qquad R_2=Y
\end{equation}

We proceed as in \cite{cluster4}. 
Define $u_n=R_{2n}$, then odd index variables can be eliminated:
$$
R_{2n+1}=u_{n+1}C u_n -1
$$
The variables $u_n$ satisfy
\begin{equation}\label{ncu}
(u_{n+2}Cu_{n+1}-1)C(u_{n+1}C u_n -1)=1+u_{n+1}^4.
\end{equation}
The initial data corresponding to \eqref{initio} becomes
\begin{equation}\label{initu}
u_0=yxy^{-1}, \qquad u_1=(1+y)x^{-1},
\end{equation}
whereas that corresponding to \eqref{initb} becomes
\begin{equation}\label{initub}
u_0=XY^{-1}X^{-1}Y(1+X)Y^{-1}, \qquad u_1=Y.
\end{equation}

\subsection{Conserved quantities and linear recursions}

We use the following expressions for the commutator $C$, expressed as a function of $u_n$ :
\begin{eqnarray}
 C&=&u_{n+1}^{-1}(u_{n+1}C u_n -1)u_{n+1}(u_{n+1}C u_n -1)^{-1}\label{firstconsa}\\
C&=&(u_{n+1}C u_n -1)^{-1}u_n(u_{n+1}C u_n -1)u_n^{-1}\label{firstconsb}
\end{eqnarray}
The first expression is obtained from  $C=R_{2n+2}^{-1}R_{2n+1}R_{2n+2}R_{2n+1}^{-1}$,
and  the second from $C=R_{2n+1}^{-1}R_{2n}R_{2n+1}R_{2n}^{-1}$.

Starting from equation \eqref{ncu}, let us substitute \ref{firstconsa} for the term $C$ in the center 
of the left hand side,
$$
(u_{n+2}Cu_{n+1}-1)u_{n+1}^{-1}(u_{n+1}C u_n -1)u_{n+1}=1+u_{n+1}^4,
$$
or
$$u_{n+2}C(u_{n+1}C u_n -1)=u_{n+1}^3+Cu_n.
$$
We conclude that
\begin{equation}\label{newnc}
u_{n+2}C=(u_{n+1}^3+Cu_{n})(u_{n+1}C u_n -1)^{-1}
\end{equation}
By Eq.\eqref{firstconsb}, we also have
\begin{equation}\label{rear}
(u_{n+1}Cu_n-1)C=u_n u_{n+1}C -1
\end{equation}
which is a quasi-commutation relation between $u_n$ and
$u_{n+1}$:
\begin{equation}\label{qcomu}
u_{n+1}Cu_n =u_n u_{n+1} +1-C^{-1}.
\end{equation}
This is to be compared with equation \eqref{qcomm} of the case $b=c=2$ above.

We can now prove the integrability of the evolution \eqref{ncu}, by finding its conserved quantity.
\begin{lemma}\label{secons}
The function
\begin{equation}\label{seconserved}
K=\big( (Cu_n)^2+(u_{n+1})^2\big)(u_{n+1}C u_n -1)^{-1}
\end{equation}
is independent of $n$.
\end{lemma}
\begin{proof}
Let
$$
K_n=u_{n+1}^{-1}(u_{n+2}C+u_{n})
$$
Using \eqref{newnc},
\begin{equation}\label{defK}
K_n=\big( (Cu_n)^2+(u_{n+1})^2\big)(u_{n+1}C u_n -1)^{-1}
\end{equation}
To prove that $K_n$ is independent of $n$, we compute
\begin{equation}\label{debut}
K_n-K_{n-1}=\big( (Cu_n)^2+(u_{n+1})^2\big)
(u_{n+1}C u_n -1)^{-1}-u_{n}^{-1}(u_{n+1}C+u_{n-1})
\end{equation}
We first need to move the factor $(u_{n+1}C u_n -1)^{-1}$ to the left of the first term. 
To do so, note that \eqref{firstconsb}, 
implies that for all $m\geq 0$
$$
(Cu_n)^m(u_{n+1}C u_n -1)^{-1}=(u_{n+1}C u_n -1)^{-1} (u_n)^m,
$$
and  \eqref{firstconsa}, 
implies that for all $m\geq 0$
$$
(u_{n+1})^m(u_{n+1}C u_n -1)^{-1}=(u_{n+1}C u_n -1)^{-1} (u_{n+1}C)^m.
$$
Applying this to the first term in \eqref{debut}:
\begin{eqnarray*}
K_n-K_{n-1}&=&(u_{n+1}C u_n -1)^{-1}\big( (u_n)^2+(u_{n+1}C)^2
-(u_{n+1}C u_n -1)u_n^{-1}(u_{n+1}C+u_{n-1}) \big)\\
&=&(u_{n+1}C u_n -1)^{-1}u_n^{-1}\big( (u_n)^3+u_{n+1}C-(u_nu_{n+1}C-1)u_{n-1}
\big)
\end{eqnarray*}
Substituting Eq. \eqref{rear} in the last term,
\begin{eqnarray*}
K_n-K_{n-1}&=&(u_{n+1}C u_n -1)^{-1}u_n^{-1}
\big((u_n)^3+u_{n+1}C-(u_{n+1}Cu_n-1)Cu_{n-1}\big) \\
&=&(u_{n+1}C u_n -1)^{-1}u_n^{-1}\big(u_n^3+Cu_{n-1}
-u_{n+1}C(u_nC u_{n-1} -1)\big) \\
&=&0
\end{eqnarray*}
as a consequence of \ref{newnc}.
The lemma follows.
\end{proof}

By the definition \eqref{defK} of $K_n$, we deduce:

\begin{thm}\label{linreconefour}
The solution $u_n$ to the system \eqref{ncu} satisfies the following 
linear recursion relations with constant coefficients:
\begin{eqnarray}\label{linonefour}
&& u_{n+2}C-u_{n+1}K+u_n=0, \\
&& u_{n+2}-Ku_{n+1}+Cu_n=0,\label{secrec}
\end{eqnarray}
where 
\begin{equation}
C=x y x^{-1}y^{-1}, \qquad K= (x^2+\big((1+y)x^{-1}\big)^2)y^{-1}\label{valcons}
\end{equation}
in the case of the initial data $(x,y)$ as in \eqref{initu}, or
\begin{equation}
C=X Y X^{-1}Y^{-1}, \qquad K= (Y^2+\big((1+X)Y^{-1}\big)^2)YX^{-1}Y^{-1}\label{valconsb}
\end{equation}
in the case of the initial data $(X,Y)$ as in \eqref{initub}.
\end{thm}
\begin{proof}
The first relation follows from the definition of $K=K_n$ \eqref{defK}.
The second follows from the first and from the quasi-commutation relation \eqref{qcomu},
\begin{eqnarray*}
u_{n+1}Ku_{n+1}-u_{n+1}( u_{n+2}+Cu_n)&=&
u_{n+2}Cu_{n+1}+u_{n}u_{n+1} -u_{n+1}( u_{n+2}+Cu_n)\\
&=&(1-C^{-1})-(1-C^{-1})=0
\end{eqnarray*}
Substituting
the initial values \eqref{initu}-\eqref{initub} into the expressions for the conserved quantities
$C$ and $K$ leads to \eqref{valcons}-\eqref{valconsb}.
\end{proof}

\subsection{(x,y) initial data: Paths with noncommutative weights and positivity}

As in Section \ref{twotwopath}, we introduce the generating function
$F(t)=\sum_{n\geq 0} t^n u_n$, and use the second linear recursion relation of Theorem
\ref{linreconefour} to compute $F$ explicitly.

\begin{thm}\label{contionefour}
The generating function $F(t)$ has the following non-commutative (finite) continued
fraction expression:
\begin{equation}\label{contfraconefour}
F(t)=\left(1-t y_1 -t^2(1-t y_3)^{-1} y_2\right)^{-1} u_0, \end{equation}
where
\begin{eqnarray*}
y_1&=&u_1 u_0^{-1}=(1+y)x^{-1}yx^{-1}y^{-1} , \\
y_2&=& (K-y_1)y_1 -C=\Big(x^2+(1+y)x^{-2}(1+y)\Big)y^{-1}x^{-1}y x^{-1}y^{-1}   , \\
y_3&=&K-y_1=\Big(x^3+(1+y)x^{-1}\Big)x^{-1}y^{-1} , \end{eqnarray*}
and $
u_0=y x y^{-1}.$
\end{thm}
\begin{proof}
The recursion relation \eqref{secrec} implies
\begin{eqnarray*}
F(t)&=& (1-t K+t^2 C)^{-1}\big(1-t(K-u_1u_0^{-1})\big) u_0\\
&=& (1-t(y_1+y_3)+t^2(y_3 y_1-y_2))^{-1}\big(1-ty_3\big) u_0\\
&=& (1-ty_1-t^2(1-t y_3)^{-1}y_2))^{-1}u_0
\end{eqnarray*}
where we have substituted $y_1=u_1u_0^{-1}$, $y_3=K-y_1$, and $y_2=y_3y_1-C$.
\end{proof}

\begin{cor}\label{coronefour}
For all $n\geq 0$, the solution $u_n$ of the system (\ref{ncu}-\ref{initu}) is a
Laurent polynomial of $x,y$ with only non-negative integer coefficients.
\end{cor}
%\begin{proof}
%The expansion of the continued fraction $F$ \eqref{contfraconefour} in powers of $t$ 
%has
%coefficients which are manifestly positive polynomials of $y_1,y_2,y_3$ with
%non-negative integer coefficients, times $u_0$.
%Substituting the expressions  \eqref{contfraconefour}
%for $y_1,y_2,y_3,u_0$ as positive Laurent polynomials of $x,y$, we deduce the 
%Corollary.
%\end{proof}

As in Section \ref{twotwopath}, the continued fraction expression \eqref{contfraconefour}
allows to interpret $R_n$ as a path partition function for all $n\geq 0$. The new feature is
that the paths involved here will be Motzkin paths of height $1$. That is, they are paths 
on the strip of $\Z^2$ delimited by $y=0$ and $y=1$, with steps $(x,y)\to (x+1,1-y)$ or 
$(x,y)\to(x+1,y)$. These are in bijection with paths on the graph $G$ of Figure 
\ref{fig:barbell}, if we consider horizontal steps at height $y$ to be steps around 
the loop connecting vertex $y$ to itself, and diagonal steps to up/down steps 
along the edge connecting vertex 0 and 1. The Motzkin paths of length $3$ are represented
in Figure \ref{fig:barbell} for illustration.

\begin{thm}\label{pathonefour}
For all $n\in \Z_{\geq 0}$, $u_nu_0^{-1}$ is the partition function for Motzkin paths of height 1, from $(0,0)$ to $(n,0)$, with weights
$y_1$ per horizontal step at height 0, weight $1$ for an upward step, weight$y_2$ per downward step and weight $y_3$ per horizontal step
at height 1, with $y_i$ as in equation \eqref{contfraconefour}.
\end{thm}
Notice that formally, these are precisely the paths of length $n$ on the weighted graph $G$ in Figure \ref{fig:barbell}.
\begin{proof}
The expansion of the continued fraction $F$ \eqref{contfraconefour} in powers of $t$ 
may be again decomposed into two steps:
\begin{eqnarray*}
F_1&=&(1-t y_3)^{-1} \\
F&=& (1-t y_1 -t(F_1)(ty_2))^{-1} u_0
\end{eqnarray*}
We may now interpret $F_1$ as the partition function for paths on $\{1\}$
made of consecutive steps $1\to 1$ (each receiving the weight $t y_3$). 
As to $F$, it generates paths on $[0,1]$ from $0$ to $0$,
made of any shuffle of steps $0\to 0$ (with weight $ty_1$)
and segments made of one step $0\to 1$ (weight $t$) followed by any number of steps $1\to 1$ 
(weight $F_1$) and then one step $1\to 0$ (weight $ty_2$).
The Theorem follows.
\end{proof}

Alternatively, we may express the partition function $u_nu_0^{-1}$ by means of the
path transfer matrix 
$T= \begin{pmatrix} y_1 & 1\\ y_2 & y_3 \end{pmatrix}$
with entries indexed $0,1$, resulting in:
\begin{equation}
u_nu_0^{-1} =\left(T^n \right)_{0,0}
\end{equation}

\begin{example}
For $n=2$ we have
\begin{eqnarray*}
u_2u_0^{-1}
&=&(u_1^3+Cu_0)(u_1Cu_0-1)^{-1}C^{-1}u_0^{-1}= \big(x+((1+y)x^{-1})^3\big) x y^{-1}x^{-1}yx^{-1}y^{-1}\\
&=& y_1^2+y_2
\end{eqnarray*}
by using the explicit values of $y_1,y_2$ from equation \eqref{contfraconefour}. 
This is the contribution of the two Motzkin paths of length $2$ on $[0,1]$ starting and ending at $0$,
namely $0\to 0\to 0$ (weight $y_1^2$) and $0\to 1\to 0$ (weight $1\times y_2$).
\end{example}

Let us now turn to the remaining variables 
$R_{2n+1}=u_{n+1}Cu_n-1=u_nu_{n+1}-C^{-1}$. From the Laurent 
positivity result for $u_n$, it is easy to deduce that of $R_{2n+1}$. We simply have
to show that the term $C^{-1}$ occurs at least once in the product $u_nu_{n+1}$.
We have

\begin{lemma}\label{containC}
For $n\geq 1$, the expression of $u_nu_{n+1}$ as a Laurent polynomial of $x,y$ contains
the term $C^{-1}$.
\end{lemma}
\begin{proof}
Using the path interpretation above, let us show that the contribution to $u_nu_{n+1}$
 of a particular pair $(m_1,m_2)$ of Motzkin paths on $[0,1]$ of lengths 
 $n$ and $n+1$ respectively also contains the term $C^{-1}$. For $m_1$ we take the flat motzkin path
of length $n$: $0\to 0\to 0\to \cdots \to 0$ (with weight $y_1^n$) and for $m_2$
the ``maximal" Motzkin path of length $n+1$: $0\to 1\to 1\to \cdots \to 1\to 0$ (with
weight $y_3^ny_2$). We are left with the task of proving that $y_1^n u_0 y_3^ny_2 u_0$,
when expressed as a Laurent polynomial of $x,y$, contains the term $C^{-1}=yxy^{-1}x^{-1}$.
In view of the explicit values of $y_1,y_2,y_3$ \eqref{contfraconefour}, let us only retain
one term in each weight, namely $\ty_i$ instead of $y_i$, with
$\ty_1=yx^{-1}yx^{-1}y^{-1}$, $\ty_2=x^2y^{-1}x^{-1}yx^{-1}y^{-1}$ and $\ty_3=x^2y^{-1}$.
We find easily that 
\begin{eqnarray*}
(\ty_1)^n u_0 (\ty_3)^{n-1}\ty_2 u_0&=&\! y x^{-1}(yx^{-2})^{n-1} yx^{-1}y^{-1}\ yxy^{-1}\ 
(x^2 y^{-1})^{n-1} x^2 y^{-1}x^{-1}yx^{-1}y^{-1}\ yxy^{-1}\\
&=&\! yxy^{-1}x^{-1}=C^{-1}
\end{eqnarray*}
and the Lemma follows.
\end{proof}

We summarize the results of this section with the following

\begin{thm}\label{posionefourpos}
For all $n\geq 0$, the solution $R_{n}$ to the system (\ref{qsysonefour}-\ref{initio})
is a Laurent polynomial of $x,y$ with only non-negative integer coefficients.
\end{thm}

\subsection{(X,Y) initial data: Paths with noncommutative weights and positivity}

Let us now re-express the generating function for $u_n$ in terms of $X,Y$.
More precisely, let us compute the new generating function
$G(t)=\sum_{n\geq 0} t^n u_{n+1}$.
We have the following

\begin{thm}\label{XY}
The generating function $G$ has the following continued fraction form:
$$
G(t)=\Big(1-t y_1'-t^2(1-t y_3')^{-1}y_2'\Big)^{-1} u_1 ,
$$
where
\begin{eqnarray*}
y_1'&=& K-y_3'=K-Cu_0u_1^{-1}=(Y^3+(1+X)Y^{-1})X^{-1}Y^{-1},  \\
y_2'&=&y_3'y_1'-C=\big(Y+(1+X)Y^{-2}(1+X)Y^{-1}\big)X^{-1}Y^{-1},  \\
y_3'&=&Cu_0u_1^{-1}=(1+X)Y^{-2}, 
\end{eqnarray*}
where $u_1= Y.$
\end{thm}
\begin{proof}
Starting from the expression for $F(t)$ of Theorem \ref{contionefour},
we compute $G(t)=(F(t)-u_0)/t$:
\begin{eqnarray}\label{othercont}
G(t)&=&{1\over t}\left( (1-t K+t^2C)^{-1}(1-t(K-u_1u_0^{-1}))-1\right)u_0 \nonumber \\
&=&(1-t K+t^2C)^{-1}(u_1u_0^{-1}-tC)u_0=(1-t K+t^2C)^{-1}(1-tCu_0u_1^{-1})u_1\nonumber \\
&=&(1-t (y_1'+y_3')+t^2(y_3'y_1'-y_2'))^{-1}(1-ty_3')u_1\nonumber \\
&=& \Big(1-t y_1'-t^2(1-t y_3')^{-1}y_2'\Big)^{-1} u_1
\end{eqnarray}
and the Theorem follows.
\end{proof}
Note that the path interpretation of the Theorem \ref{pathonefour}
still holds for $G$, but with the new weights
$y_1',y_2',y_3'$ and $u_1$ instead of $u_0$.
Noting moreover that the weights $y_1',y_2',y_3',u_1$ of equation \eqref{othercont}
are all positive Laurent polynomials
of $X,Y$, we deduce:

\begin{cor}\label{useco}
For all $n\geq 1$, the solution $u_n$ to the system (\ref{ncu}-\ref{initub})
is a Laurent polynomial of $X,Y$ with only non-negative integer coefficients.
\end{cor}

We now turn to the remaining variables $R_{2n+1}=u_nu_{n+1}-C^{-1}$.
We may repeat the analysis of Lemma \ref{containC}
in terms of the variables $X,Y$, by use of Theorem \ref{XY}. The result is

\begin{lemma}
The solution $R_{2n+1}$ to the system (\ref{qsysonefour}-\ref{initb})
has a positive Laurent polynomial expression in terms of $X,Y$. 
\end{lemma}
\begin{proof}
We must show that $u_nu_{n+1}$ contains at least once the term $C^{-1}$.
We repeat the analysis in the proof of Lemma \ref{containC}, using
the continued fraction of Theorem \ref{XY}. We pick the 
contribution of the same two 
paths to $u_nu_{n+1}$, but we now retain in the weights only the terms
$\ty_1'=Y^3X^{-1}Y^{-1}$, $\ty_2'=XY^{-2}XY^{-1}X^{-1}Y^{-1}$ and
$\ty_3'=XY^{-2}$. The contribution is then easily computed to be
\begin{eqnarray*}
(\ty_1')^n u_1 (\ty_3')^{n-1}\ty_2' u_1&=&Y(Y^2X^{-1})^n Y^{-1} Y (XY^{-2})^{n-1}
XY^{-2}XY^{-1}X^{-1}Y^{-1} Y\\
&=& YXY^{-1}X^{-1}=C^{-1}
\end{eqnarray*}
So the subtracted expression 
$R_{2n+1}=u_nu_{n+1}-C^{-1}$ is a positive Laurent polynomial of $X,Y$.
\end{proof}

We summarize the results of this section with the following

\begin{thm}\label{opsionefourpos}
For all $n\geq 0$, the solution $R_{n}$ to the system (\ref{qsysonefour}-\ref{initb})
is a Laurent polynomial of $X,Y$ with only non-negative integer coefficients.
\end{thm}

\subsection{Main theorem and the case $(b,c)=(4,1)$}

We conclude with our main theorem:

\begin{thm}\label{onefourpositivity}
For all $n\in \Z$, the solution $R_n$ of the system \eqref{qsysonefour}
for $(b,c)=(1,4)$,
with respectively initial data $(x,y)$ \eqref{initio} and initial data $(X,Y)$ \eqref{initb}
is a positive Laurent polynomial of respectively $x,y$ and $X,Y$,
with only non-negative integer coefficients. The same holds for the system with $(b,c)=(4,1)$
as well.
\end{thm}
\begin{proof}
By Theorems \ref{posionefourpos} and \ref{opsionefourpos}, we deduce that both $f^{(1,4)}_n(x,y)$
and $g^{(1,4)}_n(X,Y)$ (defined in Section \ref{symms})
are positive Laurent polynomials for all $n\geq 0$. 
By Theorem \ref{bctrans}, we deduce that $f^{(4,1)}_{n-1}(x,y)=g^{(1,4)}_{n}(x,y)$
and $g^{(4,1)}_{n+1}(X,Y)=f^{(1,4)}_{n}(X,Y)$ are also 
positive Laurent polynomials for all $n\geq 0$. Finally, by Equation \eqref{symthm}, we deduce that
both 
$f_{-n}^{(1,4)}(x,y)=\left(f_{n+1}^{(4,1)}(x,y)\right)^*$
and 
$f_{-n}^{(4,1)}(x,y)=\left(f_{n+1}^{(1,4)}(x,y)\right)^*$
are positive Laurent polynomials for all $n\geq 0$.
We then apply again Theorem \ref{bctrans} to conclude that both
$g_{-n}^{(1,4)}(X,Y)=f_{-n-1}^{(4,1)}(X,Y)$ and $g_{-n}^{(4,1)}(X,Y)=f_{-n-1}^{(1,4)}(X,Y)$
are positive Laurent polynomials for all $n\geq 0$. The Theorem follows.
\end{proof}

\end{document}